%% file: tc.tex
\documentclass{llncs}

\usepackage{amsmath}
\usepackage{amssymb}
\usepackage{stmaryrd}
\usepackage{url}
\usepackage{enumerate}

\newcommand*{\kw}[1]{\mathop{\textup{\textbf{#1}}}}
\newcommand*{\arrayrange}[2]{\mathop{[{#1}\,.\,.\,{#2}]}}

\newcommand{\summary}[1]{\textrm{\textbf{\textup{#1}}}}

\newcommand{\bigO}{\mathop{\mathrm{O}}\nolimits}

\newcommand{\bigTheta}{\mathop{\mathrm{\Theta}}\nolimits}

\providecommand*{\Zset}{\mathbb{Z}}            
\providecommand*{\Qset}{\mathbb{Q}}            

\newcommand*{\cC}{\ensuremath{\mathcal{C}}}

\newcommand*{\cN}{\ensuremath{\mathcal{N}}}

\newcommand*{\cV}{\ensuremath{\mathcal{V}}}

\newcommand*{\rS}{\ensuremath{\mathrm{S}}}
\newcommand*{\rs}{\ensuremath{\mathrm{s}}}
\newcommand*{\rT}{\ensuremath{\mathrm{T}}}
\newcommand*{\rt}{\ensuremath{\mathrm{t}}}

\newcommand*{\lowrS}{\text{\raisebox{-.2ex}{\rS}}}
\newcommand*{\lowrT}{\text{\raisebox{-.2ex}{\rT}}}

\renewcommand{\emptyset}{\mathord{\varnothing}}

\newcommand*{\union}{\cup}

\newcommand*{\fund}[3]{\mathord{#1}\colon#2\rightarrow#3}

\newcommand{\itc}{\mathrel{:}}

\spnewtheorem*{delayedproof}{Proof}{\bfseries}{\rmfamily}

\newcommand{\defeq}{\mathrel{\mathord{:}\mathord{=}}}

\newcommand*{\Vpm}{\cN}

\newcommand*{\length}[1]{\left\| #1 \right\|}
\newcommand{\pathconc}{\mathrel{::}}

\renewcommand*{\bar}[1]{\overline{#1}}
\newcommand*{\bari}{\overline{\imath}}
\newcommand*{\barj}{\overline{\jmath}}

\newcommand*{\barz}{\overline{z}}

\newcommand*{\Graphs}{\mathbb{G}}

\newcommand*{\graphleq}{\unlhd}
\newcommand*{\graphlt}{\lhd}
\newcommand*{\graphglb}{\sqcap}

\newcommand*{\graphlub}{\sqcup}
\newcommand*{\biggraphlub}{\bigsqcup}
\newcommand*{\Octgraphs}{\mathbb{O}}

\newcommand*{\closure}{\mathop{\mathrm{closure}}\nolimits}
\newcommand*{\strongclosure}{\mathop{\text{\rm S-closure}}\nolimits}
\newcommand*{\tightclosure}{\mathop{\text{\rm T-closure}}\nolimits}

\input prooftree
\begin{document}

\title{An Improved Tight Closure Algorithm\\
       for Integer Octagonal Constraints\thanks{This
work has been partly supported by MURST project
``AIDA --- Abstract Interpretation: Design and Applications,''
and by a Royal Society (UK) International Joint Project (ESEP) award.}}
\author{Roberto Bagnara\inst{1}
\and Patricia M. Hill\inst{2}
\and Enea Zaffanella\inst{1}
}
\authorrunning{R.~Bagnara,
               P.~M.~Hill,
               E.~Zaffanella
}
\tocauthor{Roberto Bagnara (University of Parma),
           Patricia M. Hill (University of Leeds),
           Enea Zaffanella (University of Parma)
}

\institute{
    Department of Mathematics,
    University of Parma,
    Italy \\
    \email{\{bagnara,%
             zaffanella\}@cs.unipr.it}
\and
    School of Computing,
    University of Leeds,
    UK \\
    \email{hill@comp.leeds.ac.uk}
}

\maketitle

\begin{abstract}
Integer octagonal constraints
(a.k.a.\ \emph{Unit Two Variables Per Inequality} or
\emph{UTVPI integer constraints})
constitute an interesting class of constraints
for the representation and solution of integer problems
in the fields of constraint programming and formal analysis
and verification of software and hardware systems,
since they couple algorithms having polynomial complexity
with a relatively good expressive power.
The main algorithms required for the manipulation of such constraints
are the satisfiability check and the computation of the inferential
closure of a set of constraints.  The latter is called \emph{tight}
closure to mark the difference with the (incomplete) closure algorithm
that does not exploit the integrality of the variables.
In this paper we present and
fully justify an $\bigO(n^3)$ algorithm to compute the tight closure
of a set of UTVPI integer constraints.
\end{abstract}

\section{Introduction}
\label{sec:introduction}

\emph{Integer octagonal constraints}, also called
\emph{Unit Two Variables Per Inequality (UTVPI) integer constraints}
---that is, constraints of the form $ax + by \leq d$
where $a, b \in \{-1, 0, +1\}$, $d \in \Zset$ and the variables $x$ and $y$
range over the integers---,
constitute an interesting subclass of linear integer constraints
admitting polynomial solvability.
The place these constraints occupy in the complexity/expressivity spectrum
is in fact peculiar.
Concerning complexity, relaxing the restriction imposing (at most)
two variables per constraint, or relaxing the restriction on coefficients,
or relaxing both restrictions make the satisfiability problem NP-complete
\cite{JaffarMSY94,Lagarias85b}.
Concerning expressivity, integer octagonal constraints can be used
for representing
and solving many integer problems in the field of constraint programming,
such as temporal reasoning and scheduling \cite{JaffarMSY94}.
In the field of formal analysis and verification of software and hardware
systems, these constraints have been successfully used in a number
of applications \cite{BalasundaramK89,BallCLZ04,CousotCFMMMR05,Mine05th}.

When (integer or rational) octagonal constraints are used to build
abstract domains\footnote{In \emph{abstract interpretation}
theory \cite{CousotC77}, an \emph{abstract domain} is an algebraic structure
formalizing  a set of approximate assertions endowed with an entailment
(or approximation) relation, plus various operations that correctly approximate
the operations of some \emph{concrete domain}, i.e., the domain being
abstracted/approximated.}
---such as the \emph{Octagon Abstract Domain} implemented in the library
with the same name \cite{Mine06b} or the domain of \emph{octagonal shapes}
defined in \cite{BagnaraHMZ05} and implemented in the
\emph{Parma Polyhedra Library} \cite{BagnaraHZ06TR}---
the most critical operation is not the satisfiability check
(although very important in constraint programming)
but \emph{closure by entailment}.
This is the procedure whereby a set of octagonal constraints
is augmented with (a finite representation of) all the
octagonal constraints that can be inferred from it.
The closure algorithms for rational octagonal constraints are sound
but not complete for integer octagonal constraints.
The latter require so-called \emph{tight} closure algorithms that
fully exploit the integrality of the variables.

In 2005, Lahiri and Musuvathi proposed an $\bigO(n^3)$ algorithm for
the satisfiability check of a (non trivially redundant) system of
UTVPI integer constraints~\cite{LahiriM05}.  They also sketched
(without formal definitions and proofs) a tight closure algorithm with
the same worst-case complexity bound.
Still in 2005, Min\'e proposed a modification of
the strong (i.e., non-tight) closure algorithm for
\emph{rational} octagonal constraints
and argued that this would provide a good
and efficient approximation of tight closure~\cite{Mine01a}.
In the same year
we showed that the algorithm for computing
the strong closure of rational octagonal constraints
as described in~\cite{Mine01a}
could be simplified with a consequential improvement in its efficiency
\cite{BagnaraHMZ05,BagnaraHMZ05TR}.
In this paper we show that our result can be extended
so as to apply to integer octagonal constraints.
This enables us to present and, for the first time, fully justify an
$\bigO(n^3)$ algorithm to compute the tight closure of a set of UTVPI
integer constraints.

In Section~\ref{sec:preliminaries} we briefly introduce the terminology
and notation adopted throughout the paper and we recall a few
standard results on weighted graphs.
In Section~\ref{sec:octagonal-graph}, we give the definition
of rational-weighted octagonal graphs and recall some of the results that
were established in~\cite{BagnaraHMZ05,BagnaraHMZ05TR}.
In Section~\ref{sec:octagonal-Z-graph}, we extend these results
to the case of integer-weighted octagonal graphs.
Finally, in Section~\ref{sec:conclusion} we conclude and briefly
discuss future work.

\section{Preliminaries}
\label{sec:preliminaries}

Let $\Qset_\infty \defeq \Qset \union \{ +\infty \}$ be totally
ordered by the extension of `$\mathord{<}$' such that $d < +\infty$ for
each $d \in \Qset$.
Let $\cN$ be a finite set of \emph{nodes}.
A \emph{rational-weighted directed graph} (graph, for short)
$G$ in $\cN$ is a pair $(\cN, w)$,
where $\fund{w}{\cN \times \cN}{\Qset_\infty}$ is the
weight function for $G$.

Let $G = (\cN, w)$ be a graph.
A pair $(n_i, n_j) \in \cN \times \cN$ is an \emph{arc} of $G$
if $w(n_i,n_j) < +\infty$;
the arc is \emph{proper} if $n_i \neq n_j$.
A \emph{path} $\pi = n_0 \cdots n_p$ in $G$
is a non-empty and finite sequence of nodes such that
$(n_{i-1}, n_i)$ is an arc of $G$, for all $i = 1$, \dots,~$p$.
Each node $n_i$ where  $i = 0$, \dots,~$p$ and
each arc $(n_{i-1}, n_i)$ where $i = 1$, \dots,~$p$
is said to be \emph{in} the path $\pi$.
The \emph{length} of the path $\pi$ is the number $p$ of occurrences
of arcs in $\pi$ and denoted by $\length{\pi}$;
the \emph{weight} of the path $\pi$
is $\sum_{i=1}^p w(n_{i-1}, n_i)$ and denoted by $w(\pi)$.
The path $\pi$ is \emph{simple} if each node occurs at most once in $\pi$.
The path $\pi$ is \emph{proper} if all the arcs in it are proper.
The path $\pi$ is a \emph{proper cycle} if it is a proper path,
$n_0 = n_p$ and $p \geq 2$.
If $\pi_1 = n_0 \cdots n_h$
and $\pi_2 = n_h \cdots n_p$ are paths,
where $0 \leq h \leq p$,
then the path concatenation $\pi = n_0 \cdots n_h \cdots n_p$ of
$\pi_1$ and $\pi_2$ is denoted by $\pi_1 \pathconc \pi_2$;
if $\pi_1 = n_0 n_1$ (so that $h = 1$),
then $\pi_1 \pathconc \pi_2$ will also be denoted by $n_0 \cdot \pi_2$.
Note that path concatenation is not the same as sequence concatenation.
The path $\pi$ is a \emph{zero-cycle}
if it is a proper cycle and $w(\pi) = 0$.
A graph is \emph{zero-cycle free} if all its proper cycles
have strictly positive weights.

A graph $(\cN, w)$ can be interpreted to represent
the system of \emph{potential constraints}
\[
  \cC
    \defeq
      \bigl\{\,
        n_i - n_j \leq w(n_i, n_j)
      \bigm|
        n_i, n_j \in \cN
      \,\bigr\}.
\]
Hence, the graph $(\cN, w)$ is \emph{consistent}
if and only if
the system of constraints it represents is satisfiable in $\Qset$,
i.e., there exists a rational valuation
$\fund{\rho}{\cN}{\Qset}$ such that,
for each constraint $(n_i - n_j \leq d) \in \cC$,
the relation $\rho(n_i) - \rho(n_j) \leq d$ holds.
It is well-known that a graph is consistent
if and only if it has no negative weight cycles
(see \cite[Section 25.5]{CormenLR90} and \cite{Pratt77}).

The set of consistent graphs in $\cN$ is denoted by $\Graphs$.
This set is partially ordered by the relation
`$\mathord{\graphleq}$' defined,
for all $G_1 = (\cN, w_1)$ and $G_2 = (\cN, w_2)$, by
\[
  G_1 \graphleq G_2
    \quad\iff\quad
      \forall i, j \in \cN \itc w_1(i,j) \leq w_2(i,j).
\]
We write $G \graphlt G'$ when $G \graphleq G'$ and $G \neq G'$.
When augmented with a bottom element $\bot$ representing inconsistency,
this partially ordered set becomes a non-complete lattice
\(
  \Graphs_\bot
    = \bigl\langle
        \Graphs \union \{ \bot \},
        \graphleq,
        \graphglb,
        \graphlub
      \bigr\rangle
\),
where `$\mathord{\graphglb}$' and `$\mathord{\graphlub}$'
denote the finitary greatest lower bound and least upper bound
operators, respectively.

\begin{definition}
\label{def:closure}
\summary{(Closed graph.)}
A consistent graph $G = (\cN, w)$ is \emph{closed}
if the following properties hold:
\begin{align}
\label{rule:zero-weight-self-loops}
  \forall i \in \cN
    &\itc w(i,i) = 0; \\
\label{rule:transitivity}
  \forall i,j,k \in \cN
    &\itc w(i,j) \leq w(i,k) + w(k,j).
\end{align}
The \emph{(shortest-path) closure} of a consistent graph $G$ in $\cN$ is
\[
  \closure(G)
    \defeq
      \biggraphlub
        \bigl\{\,
          G' \in \Graphs
        \bigm|
          \text{$G' \graphleq G$ and $G'$ is closed\/}
        \,\bigr\}.
\]
\end{definition}
When trivially extended so as to behave as the identity function
on the bottom element $\bot$, shortest-path closure is a kernel operator
(monotonic, idempotent and reductive) on the lattice $\Graphs_\bot$,
therefore providing a canonical form.

The following lemma recalls a well-known result for closed graphs
(for a proof, see Lemma~5 in~\cite{BagnaraHMZ05TR}).

\begin{lemma}
\label{lem:trans-rule-star}
Let $G =(\cN, w) \in \Graphs$ be a closed graph.
Then, for any path $\pi = i \cdots j$ in $G$,
it holds that $w(i, j) \leq w(\pi)$.
\end{lemma}

\section{Rational Octagonal Graphs}
\label{sec:octagonal-graph}

We assume in the following that there is a fixed set
$\cV = \{ v_0, \ldots, v_{n-1} \}$ of $n$ variables.
The octagon abstract domain allows for the manipulation of
\emph{octagonal constraints}
of the form  $a v_i + b v_j \leq d$, where $a,b \in \{-1,0,+1\}$,
$a \neq 0$, $v_i,v_j \in \cV$, $v_i \neq v_j$ and $d \in \Qset$.
Octagonal constraints can be encoded using potential constraints
by splitting each variable $v_i$ into two forms:
a positive form $v_i^+$, interpreted as $+v_i$;
and a negative form $v_i^-$, interpreted as $-v_i$.
Then any octagonal constraint $a v_i + b v_j \leq d$ can be written
as a potential constraint $v - v' \leq d_0$
where $v, v' \in \{v_i^+, v_i^-, v_j^+, v_j^-\}$ and $d_0 \in \Qset$.
Namely, an octagonal constraint such as $v_i + v_j \leq d$ can be
translated into the potential constraint
$v_i^+ - v_j^- \leq d$; alternatively, the same octagonal constraint
can be translated into $v_j^+ - v_i^- \leq d$.
Furthermore, unary (octagonal) constraints such as
$v_i \leq d$ and $-v_i \leq d$
can be encoded as $v_i^+ - v_i^- \leq 2d$
and $v_i^- - v_i^+ \leq 2d$, respectively.

From now on, we assume that the set of nodes
is $\Vpm \defeq \{ 0, \ldots, 2n-1 \}$.
These will denote the positive and negative forms of the variables in $\cV$:
for all $i \in \Vpm$, if $i = 2k$, then
$i$ represents the positive form $v_k^+$ and, if $i = 2k+1$, then
$i$ represents the negative form $v_k^-$ of the variable $v_k$.
To simplify the presentation, for each $i \in \Vpm$,
we let $\bari$ denote $i+1$, if $i$ is even,
and $i-1$, if $i$ is odd, so that, for all $i \in \Vpm$, we also have
$\bari \in \Vpm$ and $\bar{\bari} = i$.
Then we can rewrite a potential constraint $v - v' \leq d$
where $v \in \{v_k^+, v_k^-\}$
and $v' \in \{v_l^+, v_l^-\}$
as the potential constraint $i - j \leq d$ in $\Vpm$ where, if
$v = v^+_k$, $i = 2k$ and, if $v = v^-_k$, $i = 2k+1$;
similarly, if
$v' = v^+_l$, $j = 2l$ and, if $v' = v^-_l$, $j = 2l+1$.

It follows from the above translations that
any finite system of octagonal constraints,
translated to a set of potential constraints in $\Vpm$ as above,
can be encoded by a graph $G$ in $\Vpm$.
In particular, any finite \emph{satisfiable} system
of octagonal constraints can be encoded by a \emph{consistent}
graph in $\Vpm$.
However, the converse does not hold
since in any valuation $\rho$ of an encoding of
a set of octagonal constraints
we must also have $\rho(i) = -\rho(\bari)$, so that
the arcs $(i,j)$ and $(\barj, \bari)$ should have the same weight.
Therefore, to encode rational octagonal constraints,
we restrict attention to consistent graphs over $\Vpm$
where the arcs in all such pairs are \emph{coherent}.

\begin{definition}
\label{def:octagonal-graph}
\summary{(Octagonal graph.)}
A (rational) \emph{octagonal graph} is
any consistent graph $G =(\Vpm, w)$ that
satisfies the coherence assumption:
\begin{equation}
\label{rule:coherence}
  \forall i,j \in \Vpm
    \itc w(i,j) = w(\barj, \bari).
\end{equation}
\end{definition}
The set $\Octgraphs$ of all octagonal graphs
(with the usual addition of the bottom element,
representing an unsatisfiable system of constraints)
is a sub-lattice of $\Graphs_\bot$, sharing the same
least upper bound and greatest lower bound operators.
Note that, at the implementation level, coherence can be
automatically and efficiently enforced by letting arc $(i,j)$
and arc $(\barj, \bari)$ share the same representation.

When dealing with octagonal graphs, one has to remember
the relation linking the positive and negative forms of variables.
A proper closure by entailment procedure
should consider, besides transitivity, the following inference rule:
\begin{equation}
\label{eq:strong-coherence-inference-rule}
  \prooftree
    i - \bari \leq d_1
      \qquad
    \barj - j \leq d_2
  \justifies
    i - j \leq \frac{d_1 + d_2}{2}
  \endprooftree
\end{equation}
Thus, the standard shortest-path closure algorithm is not enough
to obtain a canonical form for octagonal graphs.

\begin{definition}
\label{def:strong-closure}
\summary{(Strongly closed graph.)}
An octagonal graph $G = (\Vpm, w)$ is \emph{strongly closed}
if it is closed and the following property holds:
\begin{equation}
\label{rule:strong-coherence}
  \forall i,j \in \Vpm
    \itc 2 w(i,j) \leq w(i, \bari) + w(\barj, j).
\end{equation}
The \emph{strong closure} of an octagonal graph $G$ in $\Vpm$ is
\[
  \strongclosure(G)
    \defeq
      \biggraphlub
        \bigl\{\,
          G' \in \Octgraphs
        \bigm|
          \text{$G' \graphleq G$ and $G'$ is strongly closed\/}
        \,\bigr\}.
\]
\end{definition}
When trivially extended to the bottom element,
strong closure is a kernel operator on the lattice of octagonal graphs.

A modified closure procedure is defined in~\cite{Mine01b},
yielding strongly closed octagonal graphs.
A significant efficiency improvement can be obtained thanks to
the following theorem
(for a proof, see Theorem~2 in~\cite{BagnaraHMZ05TR}).

\begin{theorem}
\label{thm:one-step-strong-closure}
Let $G = (\Vpm, w)$ be a closed octagonal graph.
Consider the graph $G_\rS = (\Vpm, w_\lowrS)$,
where $w_\lowrS$ is defined, for each $i,j \in \Vpm$, by
\[
  w_\lowrS(i, j)
    \defeq
      \min
        \biggl\{
          w(i, j),
          \frac{w(i, \bari)}{2} + \frac{w(\barj, j)}{2}
        \biggr\}.
\]
Then $G_\rS = \strongclosure(G)$.
\end{theorem}
Intuitively, the theorem states that strong closure can be obtained
by application of any shortest-path closure algorithm
followed by a \emph{single} local propagation step using the
constraint inference rule~\eqref{eq:strong-coherence-inference-rule}.
In contrast, in the strong closure algorithm of~\cite{Mine01b},
the outermost iterations of (a variant of) the Floyd-Warshall
shortest-path algorithm are interleaved with $n$ applications
of the inference rule~\eqref{eq:strong-coherence-inference-rule},
leading to a more complex and less efficient implementation.

\section{Integer Octagonal Graphs}
\label{sec:octagonal-Z-graph}
We now consider the case of integer octagonal constraints, i.e.,
octagonal constraints where the bounds are all integral and the variables
are only allowed to take integral values.
These can be encoded by suitably restricting the codomain of the
weight function of octagonal graphs.

\begin{definition}
\label{def:integer-octagonal-graph}
\summary{(Integer octagonal graph.)}
An \emph{integer octagonal graph} is an octagonal graph $G = (\Vpm, w)$
having an integral weight function:
\[
  \forall i,j \in \Vpm \itc w(i,j) \in \Zset \union \{+\infty\}.
\]
\end{definition}

As an integer octagonal graph is also a rational octagonal graph,
the constraint system it encodes
will be satisfiable when interpreted to take values in $\Qset$.
However, when interpreted to take values in $\Zset$,
this system may be unsatisfiable
since the arcs encoding unary constraints can have an odd weight;
we say that an octagonal graph is \emph{$\Zset$-consistent}
if its encoded integer constraint system is satisfiable.
For the same reason, the strong closure of an integer octagonal graph does
not provide a canonical form for the integer constraint system it encodes
and we need to consider the following \emph{tightening} inference rule:
\begin{equation}
\label{eq:tightening-inference-rule}
  \prooftree
    i - \bari \leq d
  \justifies
    i - \bari \leq 2 \lfloor d/2 \rfloor
  \endprooftree.
\end{equation}

\begin{definition}
\label{def:tight-closure}
\summary{(Tightly closed graph.)}
An octagonal graph $G = (\Vpm, w)$ is \emph{tightly closed}
if it is a strongly closed integer octagonal graph
and the following property holds:
\begin{equation}
\label{rule:tight-coherence}
  \forall i \in \Vpm
    \itc \mathord{\text{$w(i, \bari)$ is even}}.
\end{equation}
The \emph{tight closure} of an octagonal graph $G$ in $\Vpm$ is
\[
  \tightclosure(G)
    \defeq
      \biggraphlub
        \bigl\{\,
          G' \in \Octgraphs
        \bigm|
          \text{$G' \graphleq G$ and $G'$ is tightly closed\/}
        \,\bigr\}.
\]
\end{definition}

By property \eqref{rule:tight-coherence},
any tightly closed integer octagonal graph
will encode a satisfiable integer constraint system
and is therefore $\Zset$-consistent.
Moreover, since the encoding of any satisfiable integer constraint system
will result in a $\Zset$-consistent integer octagonal graph $G$
that satisfies property \eqref{rule:tight-coherence},
its tight closure $\tightclosure(G)$ will also be $\Zset$-consistent.
This means that, if $G$ is \emph{not} $\Zset$-consistent,
then $\tightclosure(G) = \biggraphlub \emptyset = \bot$;
that is, the tight closure operator computes
either a tightly closed graph or the bottom element.
Therefore, tight closure is a kernel operator on the lattice of
octagonal graphs, as was the case for strong closure.

An incremental closure procedure for obtaining the tight closure of an
octagonal graph was defined in~\cite{JaffarMSY94} and improved
in~\cite{HarveyS97}.
The algorithm, which is also presented and discussed
in~\cite[Section~4.3.5]{Mine05th}, maintains the tight closure
of a system of octagonal constraints by performing at most
$\bigO(n^2)$ operations each time a new constraint is added:
thus, for $m$ constraints, the worst case complexity is $\bigO(m n^2)$.
In particular, for the case of a dense system of octagonal constraints
where $m \in \bigO(n^2)$, the worst case complexity is $\bigO(n^4)$.

The following theorem shows that a more efficient tight closure algorithm
can be obtained by a simple modification to the improved strong closure
algorithm of Theorem~\ref{thm:one-step-strong-closure}.
Basically,
inference rule~\eqref{eq:tightening-inference-rule}
must be applied to ensure property \eqref{rule:tight-coherence}
holds before applying
inference rule~\eqref{eq:strong-coherence-inference-rule}.

\begin{theorem}
\label{thm:one-step-tight-closure}
Let $G = (\Vpm, w)$ be a closed integer octagonal graph.
Consider the graph $G_\rT = (\Vpm, w_\lowrT)$,
where $w_\lowrT$ is defined, for each $i,j \in \Vpm$, by
\[
  w_\lowrT(i, j)
    \defeq
      \min
        \biggl\{
          w(i, j),
          \Bigl\lfloor \frac{w(i, \bari)}{2} \Bigr\rfloor
            + \Bigl\lfloor \frac{w(\barj, j)}{2} \Bigr\rfloor
        \biggr\}.
\]
Then, if $G_\rT$ is an octagonal graph, $G_\rT = \tightclosure(G)$.
\end{theorem}

\begin{figure}
\begin{align*}
& \kw{procedure} \;
    \texttt{tight\_closure\_if\_consistent}
      (\kw{var} \; w\arrayrange{0}{2n-1}\arrayrange{0}{2n-1}) \\
& \qquad
  \{ \text{ Classical Floyd-Warshall: $\bigO(n^3)$ } \} \\
& \qquad
    \kw{for} \; k := 0 \; \kw{to} \; 2n-1 \; \kw{do} \\
& \qquad\qquad
    \kw{for} \; i := 0 \; \kw{to} \; 2n-1 \; \kw{do} \\
& \qquad\qquad\qquad
      \kw{for} \; j := 0 \; \kw{to} \; 2n-1 \; \kw{do} \\
& \qquad\qquad\qquad\qquad
        w[i,j] := \min\bigl(
                        w[i,j],
                        w[i,k] + w[k,j]
                      \bigr); \\
& \qquad
  \{ \text{ Tight coherence: $\bigO(n^2)$ } \} \\
& \qquad
  \kw{for} \; i := 0 \; \kw{to} \; 2n-1 \; \kw{do} \\
& \qquad\qquad
    \kw{for} \; j := 0 \; \kw{to} \; 2n-1 \; \kw{do} \\
& \qquad\qquad\qquad
      w[i,j] := \min\Bigl(
                      w[i,j],
                      \mathrm{floor}\bigl(w[i,\bari]/2\bigr)
                        + \mathrm{floor}\bigl(w[\barj,j]/2\bigr)
                    \Bigr);
\end{align*}
\caption{A $\bigO(n^3)$ tight closure algorithm for $\Zset$-consistent
integer octagonal graphs}
\label{fig:consistent-tight-closure}
\end{figure}
Figure~\ref{fig:consistent-tight-closure}
shows the pseudo-code for a $\bigO(n^3)$ tight closure algorithm
based on Theorem~\ref{thm:one-step-tight-closure} and
on the classical Floyd-Warshall shortest-path closure algorithm.
Note that the pseudo-code in Figure~\ref{fig:consistent-tight-closure}
assumes that the data structure recording the weight function $w$,
here denoted to be similar to a bidimensional array,
automatically implements the coherence assumption for octagonal graphs
(i.e., property~\eqref{rule:coherence}
of Definition~\ref{def:octagonal-graph}).

In the case of sparse graphs, a better complexity bound can be
obtained by modifying the code in Figure~\ref{fig:consistent-tight-closure}
so as to compute the shortest path closure
using Johnson's algorithm~\cite{CormenLR90}:
the worst case complexity of such an implementation will be
$\bigO(n^2 \log n + mn)$,
which significantly improves upon the $\bigO(mn^2)$
worst case complexity of~\cite{HarveyS97,JaffarMSY94}
when, e.g., $m \in \bigTheta(n)$.
However, as observed elsewhere~\cite{Mine05th,VenetB04},
some of the targeted applications (e.g., static analysis)
typically require the computation of graphs that are dense,
so that the Floyd-Warshall algorithm is often a better choice
from a practical perspective.

It is possible to define an incremental variant of the tight closure
algorithm in Figure~\ref{fig:consistent-tight-closure}, which is simply
based on the corresponding incremental version of the Floyd-Warshall
shortest path closure algorithm. In such a case, we obtain the same
worst case complexity of~\cite{HarveyS97,JaffarMSY94}.

The proof of Theorem~\ref{thm:one-step-tight-closure}
relies on a few auxiliary lemmas.
The first two were also used in~\cite{BagnaraHMZ05TR}
for the formal proof of Theorem~\ref{thm:one-step-strong-closure} above
(for their detailed proofs, see Lemmas~9 and 10 in~\cite{BagnaraHMZ05TR}).

\begin{lemma}
\label{lem:ijarc-octagonal->closed}
Let $G =(\Vpm, w)$ be an octagonal graph,
$G^\star = (\Vpm, w^\star) \defeq \closure(G)$
and $(z_1, z_2)$ be an arc in $G^\star$.
Then there exists a simple path
$\pi = z_1 \cdots z_2$ in $G$ such that
$w^\star(z_1, z_2) = w(\pi)$.
\end{lemma}

\begin{lemma}
\label{lem:closed->1-step-strong-coherence}
Let $G =(\Vpm, w)$ be a closed octagonal graph and
$i, j \in \Vpm$ be such that $i \neq \barj$ and
$2 w(i, j) \geq w(i, \bari) + w(\barj, j)$.
Let $G_\rs^\star = (\Vpm, w_\rs^\star) \defeq \closure(G_\rs)$
where $G_\rs \defeq (\Vpm, w_\rs)$ and, for each $h_1, h_2 \in \Vpm$,
\begin{align*}
  w_\rs(h_1, h_2)
    &\defeq
      \begin{cases}
        \bigl(w(i, \bari) + w(\barj, j)\bigr)/2,
          &\text{if \(
                      (h_1, h_2)
                        \in \bigl\{ (i, j), (\barj, \bari) \bigr\}
                    \);} \\
        w(h_1, h_2),
          &\text{otherwise.}
      \end{cases}
\end{align*}
Let also $z_1, z_2 \in \Vpm$.
Then one or both of the following hold:
\begin{align*}
  w_\rs^\star(z_1, z_2)
    &= w(z_1, z_2); \\
  2 w_\rs^\star(z_1, z_2)
    &\geq w(z_1, \barz_1) + w(\barz_2, z_2).
\end{align*}
\end{lemma}
Informally, Lemma~\ref{lem:closed->1-step-strong-coherence}
states that if inference rule~\eqref{eq:strong-coherence-inference-rule}
is applied
to a closed octagonal graph, then the resulting graph
 can be closed just by making further
applications of inference rule~\eqref{eq:strong-coherence-inference-rule}.
Note that, if $G$ is an integer octagonal graph
and property \eqref{rule:tight-coherence} holds,
then the derived graph $G_\rs$ will also be an integer octagonal graph.
We now state a new lemma for integer octagonal graphs showing that when
inference rule~\eqref{eq:tightening-inference-rule} is applied we
obtain a similar conclusion to that for
Lemma~\ref{lem:closed->1-step-strong-coherence}.

\begin{lemma}
\label{lem:closed->1-step-tight-coherence}
Let $G =(\Vpm, w)$ be a closed integer octagonal graph
and $i \in \Vpm$.
Let $G_\rt^\star \defeq \closure(G_\rt)$
where $G_\rt \defeq (\Vpm, w_\rt)$ is an octagonal graph
and, for each $h_1, h_2 \in \Vpm$,
\begin{align}
\label{eq:closed->1-step-tight-coherence:os-def}
  w_\rt(h_1, h_2)
    &\defeq
      \begin{cases}
        w(i, \bari) -1,
          &\text{if $(h_1, h_2) = (i, \bari)$;} \\
        w(h_1, h_2),
          &\text{otherwise.}
      \end{cases}
\end{align}
Let $G_\rt^\star = (\Vpm, w_\rt^\star)$
and $z_1, z_2 \in \Vpm$.
Then one or both of the following hold:
\begin{align}
\label{eq:closed->1-step-tight-coherence:simple-eq}
  w_\rt^\star(z_1, z_2)
    &= w(z_1, z_2), \\
\label{eq:closed->1-step-tight-coherence:tight-coherence-ineq}
  w_\rt^\star(z_1, z_2)
    &\geq
      \Bigl\lfloor \frac{w(z_1, \barz_1)}{2} \Bigr\rfloor
        + \Bigl\lfloor \frac{w(\barz_2, z_2)}{2} \Bigr\rfloor.
\end{align}
\end{lemma}
\begin{proof}
By hypothesis and Definition~\ref{def:closure},
$G_\rt^\star \graphleq G_\rt \graphleq G$.
If $(z_1, z_2)$ is not an arc in $G_\rt^\star$,
then $w_\rt^\star(z_1, z_2) = +\infty$;
thus, as $G_\rt^\star \graphleq G$, we also have
$w(z_1, z_2) = +\infty$
and hence property~\eqref{eq:closed->1-step-tight-coherence:simple-eq} holds.
Suppose now that $(z_1, z_2)$ is an arc in $G_\rt^\star$.
Then we can apply Lemma~\ref{lem:ijarc-octagonal->closed},
so that there exists a simple path
$\pi = z_1 \cdots z_2$ in $G_\rt$ such that
$w_\rt^\star(z_1, z_2) = w_\rt(\pi)$.

Suppose first that $w_\rt(\pi) = w(\pi)$.
Then, as $G$ is closed,
by Lemma~\ref{lem:trans-rule-star} we obtain
$w(\pi) \geq w(z_1, z_2)$
so that $w_\rt^\star(z_1, z_2) \geq w(z_1, z_2)$.
However $G_\rt^\star \graphleq G$ so that
$w_\rt^\star(z_1, z_2) \leq w(z_1, z_2)$ and therefore
property~\eqref{eq:closed->1-step-tight-coherence:simple-eq} holds.

Secondly, suppose that $w_\rt(\pi) \neq w(\pi)$.
Then, by Equation~\eqref{eq:closed->1-step-tight-coherence:os-def},
$(i, \bari)$ must be an arc in $\pi$, so that
\begin{equation}
\label{eq:closed->1-step-tight-coherence:only-ij}
  \pi = \pi_1 \pathconc (i \, \bari) \pathconc \pi_2,
\end{equation}
where
$\pi_1 = z_1 \cdots i$,
$\pi_2 = j \cdots z_2$
are simple paths in $G_\rt$ that do not contain the arc $(i, \bari)$.
Therefore,
by Equation~\eqref{eq:closed->1-step-tight-coherence:os-def},
we have $w_\rt(\pi_1) = w(\pi_1)$, $w_\rt(\pi_2) = w(\pi_2)$.

Consider~\eqref{eq:closed->1-step-tight-coherence:only-ij} and let%
\footnote{%
If $\pi = j_0 \cdots j_p$ is a path in a graph in $\Vpm$,
then $\bar{\pi}$ denotes the path $\barj_p \cdots \barj_0$.}
\begin{align*}
  \pi_1' &= \pi_1 \pathconc (i \, \bari) \pathconc \bar{\pi}_1,
  &\pi_2' &= \bar{\pi}_2 \pathconc (i \, \bari) \pathconc \pi_2. \\
\intertext{%
As $G$ is an octagonal graph, we have
$w(\pi_1) = w(\bar{\pi}_1)$
and $w(\pi_2) = w(\bar{\pi}_2)$ so that
}
  w(\pi_1') &= 2 w(\pi_1) + w(i, \bari),
  &w(\pi_2') &= 2 w(\pi_2) + w(i, \bari).\\
\intertext{%
As $G$ is closed, by Lemma~\ref{lem:trans-rule-star},
}
  w(\pi_1') &\geq w(z_1, \barz_1),
  &w(\pi_2') &\geq w(\barz_2, z_2)
\intertext{%
so that
}
  w(\pi_1) + \frac{w(i, \bari)}{2}
    &\geq
      \frac{w(z_1, \barz_1)}{2},
  &w(\pi_2) + \frac{w(i, \bari)}{2}
    &\geq
      \frac{w(\barz_2, z_2)}{2}.
\end{align*}
Therefore
\begin{align*}
  w_\rt(\pi)
    &= w_\rt(\pi_1) + w_\rt(i, \bari) + w_\rt(\pi_2) \\
    &= w(\pi_1) + \frac{w(i, \bari) - 1}{2}
         + w(\pi_2) + \frac{w(i, \bari) - 1}{2} \\
    &\geq
      \frac{w(z_1, \barz_1)}{2} - \frac{1}{2}
        + \frac{w(\barz_2, z_2)}{2} - \frac{1}{2} \\
    &\geq
      \Bigl\lfloor \frac{w(z_1, \barz_1)}{2} \Bigr\rfloor
        + \Bigl\lfloor \frac{w(\barz_2, z_2)}{2} \Bigr\rfloor.
\end{align*}
Hence, as $w_\rt^\star(z_1, z_2) = w_\rt(\pi)$, we obtain
property~\eqref{eq:closed->1-step-tight-coherence:tight-coherence-ineq},
as required.
\qed
\end{proof}

The next result, uses Lemmas~\ref{lem:closed->1-step-strong-coherence}
and~\ref{lem:closed->1-step-tight-coherence} to derive a property relating the weight
functions for a closed integer octagonal graph and its tight closure.

\begin{lemma}
\label{lem:ijarc-closed->tightclosed}
Let $G =(\Vpm, w)$ be a closed integer octagonal graph
such that $G^\rT = (\Vpm, w^\rT) \defeq \tightclosure(G)$
is an octagonal graph and let $z_1, z_2 \in \Vpm$.
Then one or both of the following hold:
\begin{align}
\label{eq:ijarc-closed->tightclosed:simple-eq}
  w^\rT(z_1, z_2)
    &= w(z_1, z_2); \\
\label{eq:ijarc-closed->tightclosed:tight-coherence-eq}
  w^\rT(z_1, z_2)
    &= \Bigl\lfloor \frac{w(z_1, \barz_1)}{2} \Bigr\rfloor
         + \Bigl\lfloor \frac{w(\barz_2, z_2)}{2} \Bigr\rfloor.
\end{align}
\end{lemma}
\begin{proof}
The proof is by contraposition; thus we assume that
neither~\eqref{eq:ijarc-closed->tightclosed:simple-eq}
nor~\eqref{eq:ijarc-closed->tightclosed:tight-coherence-eq} hold.
Without loss of generality,
let the graph $G$ be $\mathord{\graphleq}$-minimal in the set
of all closed integer octagonal graphs such that $\tightclosure(G) = G^\rT$
and for which neither~\eqref{eq:ijarc-closed->tightclosed:simple-eq}
nor~\eqref{eq:ijarc-closed->tightclosed:tight-coherence-eq} hold.
Clearly the negation of~\eqref{eq:ijarc-closed->tightclosed:simple-eq}
implies that $G \neq G^\rT$, so that $G^\rT \graphlt G$.

As $G$ is closed but not tightly closed,
by Definitions~\ref{def:strong-closure} and~\ref{def:tight-closure},
it follows that there exist $i,j \in \Vpm$ such that either
\begin{enumerate}[(i)]
\item
\label{item:ijarc-closed->tightclosed:tight-coherence}
$i = \barj$ and $w(i, \bari)$ is odd; or
\item
\label{item:ijarc-closed->tightclosed:strong-coherence}
property \eqref{rule:tight-coherence} holds and
$2 w(i, j) > w(i, \bari) + w(\barj, j)$.
\end{enumerate}
Consider graph $G_1 = (\Vpm, w_1)$
where the weight function $w_1$ is defined,
for all $h_1,h_2 \in \Vpm$, by
\begin{equation*}
  w_1(h_1, h_2)
    \defeq
      \begin{cases}
        \Bigl\lfloor \frac{w(i, \bari)}{2} \Bigr\rfloor
          + \Bigl\lfloor \frac{w(\barj, j)}{2} \Bigr\rfloor,
          &\text{if \(
                      (h_1, h_2)
                        \in \bigl\{ (i, j), (\barj, \bari) \bigr\}
                    \);} \\
        w(h_1, h_2),
          &\text{otherwise.}
      \end{cases}
\end{equation*}

Let $G_1^\star = \closure(G_1)$.
By Definitions~\ref{def:closure},~\ref{def:strong-closure}
and~\ref{def:tight-closure},
\begin{equation}
\label{eq:ijarc-closed->tightclosed:graph-relations}
G^\rT \graphleq G_1^\star \graphleq G_1 \graphlt G.
\end{equation}
Thus $\tightclosure(G_1^\star) = G^\rT$ so that,
by the minimality assumption on $G$,
one or both of the following hold:
\begin{align}
\label{eq:ijarc-closed->tightclosed:simple-eq-osc}
  w^\rT(z_1, z_2) &= w_1^\star(z_1, z_2);\\
\label{eq:ijarc-closed->tightclosed:tight-coherence-eq-osc}
  w^\rT(z_1, z_2)
    &= \Bigl\lfloor \frac{w_1^\star(z_1, \barz_1)}{2} \Bigr\rfloor
         + \Bigl\lfloor \frac{w_1^\star(\barz_2, z_2)}{2} \Bigr\rfloor.
\end{align}

As $G^\rT \neq \bot$,
by \eqref{eq:ijarc-closed->tightclosed:graph-relations}, $G_1$ is consistent.
Therefore, by construction, $G_1$ is an integer octagonal graph.
If property~(\ref{item:ijarc-closed->tightclosed:tight-coherence})
holds for $i,j$, then  Lemma~\ref{lem:closed->1-step-tight-coherence}
can be applied and,
if property~(\ref{item:ijarc-closed->tightclosed:strong-coherence})
holds for $i,j$,
then  Lemma~\ref{lem:closed->1-step-strong-coherence} can be applied
and also,
since property \eqref{rule:tight-coherence} holds,
both $w_1(z_1, \barz_1)$ and $w(\barz_2, z_2)$ are even.
Hence, letting $G_1^\star \defeq (\Vpm, w_1^\star)$,
one or both of the following hold:
\begin{align}
\label{eq:ijarc-closed->tightclosed:simple-eq-osc-2}
  w_1^\star(z_1, z_2)
    &= w(z_1, z_2); \\
\label{eq:ijarc-closed->tightclosed:tight-coherence-eq-osc-2}
  w_1^\star(z_1, z_2)
    &\geq
      \Bigl\lfloor \frac{w(z_1, \barz_1)}{2} \Bigr\rfloor
        + \Bigl\lfloor \frac{w(\barz_2, z_2)}{2} \Bigr\rfloor. \\
\intertext{%
Again by Lemmas~\ref{lem:closed->1-step-strong-coherence}
and~\ref{lem:closed->1-step-tight-coherence},
}
\notag
  w_1^\star(z_1, \barz_1)
    &\geq
      2 \Bigl\lfloor \frac{w(z_1, \barz_1)}{2} \Bigr\rfloor, \\
\notag
  w_1^\star(\barz_2, z_2)
    &\geq
      2 \Bigl\lfloor \frac{w(\barz_2, z_2)}{2} \Bigr\rfloor;
\end{align}
since the lower bounds for $w_1^\star(z_1, \barz_1)$
and $w_1^\star(\barz_2, z_2)$ are even integers, we obtain
\begin{equation}
\label{eq:ijarc-closed->tightclosed:tight-coherence-eq-osc-3}
  \Bigl\lfloor \frac{w_1^\star(z_1, \barz_1)}{2} \Bigr\rfloor
    + \Bigl\lfloor \frac{w_1^\star(\barz_2, z_2)}{2} \Bigr\rfloor
    \geq
      \Bigl\lfloor \frac{w(z_1, \barz_1)}{2} \Bigr\rfloor
        + \Bigl\lfloor \frac{w(\barz_2, z_2)}{2} \Bigr\rfloor.
\end{equation}

Suppose first that~\eqref{eq:ijarc-closed->tightclosed:simple-eq-osc}
and~\eqref{eq:ijarc-closed->tightclosed:simple-eq-osc-2} hold.
Then by transitivity
we obtain~\eqref{eq:ijarc-closed->tightclosed:simple-eq},
contradicting the contrapositive assumption for $G$.

If~\eqref{eq:ijarc-closed->tightclosed:simple-eq-osc}
and~\eqref{eq:ijarc-closed->tightclosed:tight-coherence-eq-osc-2} hold,
then it follows
\begin{equation}
\label{eq:ijarc-closed->tightclosed:tight-coherence-eq-osc-4}
  w^\rT(z_1, z_2)
    \geq
      \Bigl\lfloor \frac{w(z_1, \barz_1)}{2} \Bigr\rfloor
        + \Bigl\lfloor \frac{w(\barz_2, z_2)}{2} \Bigr\rfloor.
\end{equation}
On the other hand,
if~\eqref{eq:ijarc-closed->tightclosed:tight-coherence-eq-osc} holds,
then, by~\eqref{eq:ijarc-closed->tightclosed:tight-coherence-eq-osc-3},
we obtain again
property~\eqref{eq:ijarc-closed->tightclosed:tight-coherence-eq-osc-4}.
However, by Definition~\ref{def:tight-closure}
we also have
\begin{align*}
  w^\rT(z_1, z_2)
    &\leq
      \Bigl\lfloor \frac{w(z_1, \barz_1)}{2} \Bigr\rfloor
        + \Bigl\lfloor \frac{w(\barz_2, z_2)}{2} \Bigr\rfloor.
\end{align*}
By combining this inequality
with~\eqref{eq:ijarc-closed->tightclosed:tight-coherence-eq-osc-4}
we obtain~\eqref{eq:ijarc-closed->tightclosed:tight-coherence-eq},
contradicting the contrapositive assumption for $G$.
\qed
\end{proof}

\begin{delayedproof}[of Theorem~\ref{thm:one-step-tight-closure}]
Let $G^\rT \defeq \tightclosure(G)$.
By definition of $G_\rT$,
$G_\rT \graphleq G$ so that
$\tightclosure(G_\rT) \graphleq G^\rT$.
As $G_\rT$ is an octagonal graph, $G_\rT$ is consistent,
and hence $G^\rT \neq \bot$;
let $G^\rT = (\Vpm, w^\rT)$.
Letting $i, j \in \Vpm$,
to prove the result we need to show that
$w^\rT(i, j) = w_\lowrT(i, j)$.
Let
\(
  k_{ij}
    \defeq
      \bigl\lfloor w(i, \bari)/2 \bigr\rfloor
        + \bigl\lfloor w(\barj, j)/2 \bigr\rfloor
\).

By Definitions~\ref{def:closure},
\ref{def:strong-closure} and~\ref{def:tight-closure},
it follows that both properties
$w^\rT(i, j) \leq w(i, j)$ and
$w^\rT(i, j) \leq k_{ij}$ hold
so that, by definition of $w_\lowrT$, we have
$w^\rT(i, j) \leq w_\lowrT(i, j)$.
By Lemma~\ref{lem:ijarc-closed->tightclosed}, $w^\rT(i, j) = w(i, j)$
and/or
$w^\rT(i, j) = k_{ij}$.
Therefore since, by definition,
$w_\lowrT(i, j) = \min \bigl\{ w(i, j), k_{ij} \bigr\}$,
we obtain $w_\lowrT(i, j) \leq w^\rT(i, j)$.
\qed
\end{delayedproof}

It follows from the statement of
Theorem~\ref{thm:one-step-tight-closure}
that an implementation based on it
also needs to check the consistency of $G_\rT$.
In principle, one could apply again a
shortest-path closure procedure so as to check whether
$G_\rT$ contains some negative weight cycles.
Fortunately, a much more efficient solution is obtained
by the following result.

\begin{theorem}
\label{thm:even-closure-consistency-condition}
Let $G = (\Vpm, w)$ be a closed integer octagonal graph.
Consider the graphs $G_\rt = (\Vpm, w_\rt)$
and $G_\rT = (\Vpm, w_\lowrT)$
where, for each $i, j \in \Vpm$,
\begin{align}
\label{eq:even-closure-consistency-condition:tightening-rule}
  w_\rt(i, j)
    &\defeq
      \begin{cases}
        2 \lfloor w(i, j)/2 \rfloor,
          &\text{if $j = \bari$;} \\
        w(i, j),
          &\text{otherwise;}
      \end{cases}\\
\label{eq:even-closure-consistency-condition:strong-coherence-rule}
  w_\lowrT(i, j)
    &\defeq
      \min
        \biggl\{
          w(i, j),
          \Bigl\lfloor \frac{w(i, \bari)}{2} \Bigr\rfloor
            + \Bigl\lfloor \frac{w(\barj, j)}{2} \Bigr\rfloor
        \biggr\}.
\end{align}
Suppose that, for all $i \in \Vpm$, $w_\rt(i, \bari) + w_\rt(\bari, i) \geq 0$.
Then $G_\rT$ is an octagonal graph.
\end{theorem}

This result is a corollary of the following result
proved in~\cite[Lemma 4]{LahiriM05}.
\begin{lemma}
\label{lem:LahiriM05-lemma-4}
Let $G =(\Vpm, w)$ be an integer octagonal graph with no negative
weight cycles and
$G_\rt = (\Vpm, w_\rt)$, where $w_\rt$ satisfies
\eqref{eq:even-closure-consistency-condition:tightening-rule},
have a negative weight cycle.
Then there exists $i,\bari \in \Vpm$ and a cycle
$\pi = (i \cdot \pi_1 \cdot \bari) \pathconc (\bari  \cdot\pi_2 \cdot i)$
in $G$ such that $w(\pi) = 0$ and the weight of the shortest path in
$G$ from $i$ to $\bari$ is odd.
\end{lemma}

\begin{delayedproof}[of Theorem~\ref{thm:even-closure-consistency-condition}]
The proof is by contradiction; suppose $G_\rT$ is not an octagonal
graph; then by Definitions~\ref{def:closure}, \ref{def:strong-closure}
and~\ref{def:tight-closure}, $G_\rT$ is inconsistent.
We show that $G_\rt$ is also inconsistent.
Again, we assume to the contrary that $G_\rt$ is consistent
and derive a contradiction.
Let $i,j \in \Vpm$.
By~\eqref{eq:even-closure-consistency-condition:tightening-rule},
we have $w_\rt(i,j) \leq w(i,j)$
and $w_\rt(i, \bari)/2 + w_\rt(\barj, j)/2 = k_{ij}$,
where
\(
  k_{ij}
    \defeq
      \bigl\lfloor w(i, \bari)/2 \bigr\rfloor
        + \bigl\lfloor w(\barj, j)/2 \bigr\rfloor
\).
Letting $\strongclosure(G_\rt) = (\Vpm, w^\rS_\rt)$,
we have, by Definition~\ref{def:strong-closure},
$w^\rS_\rt(i,j) \leq w_\rt(i, j)$
and
\(
  w^\rS_\rt(i,j)
    \leq
      w_\rt(i, \bari)/2 + w_\rt(\barj, j)/2.
\)
Thus
\(
  w^\rS_\rt(i,j) \leq \min\bigl(w(i,j), k_{ij}\bigr).
\)
As this holds for all $i,j \in \Vpm$,
by~\eqref{eq:even-closure-consistency-condition:strong-coherence-rule},
$\strongclosure(G_\rt) \graphleq G_\rT$,
contradicting the assumption that $G_\rt$ was consistent.
Hence $G_\rt$ is inconsistent and
therefore contains a negative weight cycle.

By Lemma \ref{lem:LahiriM05-lemma-4},
there exists $i,\bari \in \Vpm$ and a cycle
$\pi = (i \cdot \pi_1 \cdot \bari) \pathconc (\bari  \cdot\pi_2 \cdot i)$
in $G$ such that $w(\pi) = 0$ and the weight of the shortest path in
$G$ from $i$ to $\bari$ is odd.
As $G$ is closed, $w(i, \bari) \leq w(i \cdot \pi_1 \cdot \bari)$
and $w(\bari, i) \leq w(\bari \cdot \pi_2 \cdot i)$.
Thus $w(i,\bari) + w(\bari, i) \leq w(\pi) = 0$.
Moreover, $(i \bari)$ is a path
and hence the shortest path from $i$ to $\bari$
so that $w(i \bari)$ is odd; hence,
by~\eqref{eq:even-closure-consistency-condition:tightening-rule},
$w(i, \bari) = w_\rt(i, \bari) + 1$ and $w(\bari, i) \geq w_\rt(\bari, i)$.
Therefore $w_\rt(i,\bari) + w_\rt(\bari, i) < 0$.
\qed
\end{delayedproof}

\begin{figure}
\begin{align*}
& \kw{function} \;
    \texttt{tight\_closure}
      (\kw{var} \; w\arrayrange{0}{2n-1}\arrayrange{0}{2n-1}) : \kw{bool} \\
& \qquad
  \{ \text{ Initialization: $\bigO(n)$ } \} \\
& \qquad
    \kw{for} \; i := 0 \; \kw{to} \; 2n-1 \; \kw{do} \; w[i,i] := 0; \\
& \qquad
  \{ \text{ Classical Floyd-Warshall: $\bigO(n^3)$ } \} \\
& \qquad
    \kw{for} \; k := 0 \; \kw{to} \; 2n-1 \; \kw{do} \\
& \qquad\qquad
    \kw{for} \; i := 0 \; \kw{to} \; 2n-1 \; \kw{do} \\
& \qquad\qquad\qquad
      \kw{for} \; j := 0 \; \kw{to} \; 2n-1 \; \kw{do} \\
& \qquad\qquad\qquad\qquad
        w[i,j] := \min\bigl(
                        w[i,j],
                        w[i,k] + w[k,j]
                      \bigr); \\
& \qquad
  \{ \text{ Check for $\Qset$-consistency: $\bigO(n)$ } \} \\
& \qquad
    \kw{for} \; i := 0 \; \kw{to} \; 2n-2 \; \kw{step} \; 2 \; \kw{do} \\
& \qquad\qquad
      \kw{if} \; w[i,i] < 0 \; \kw{return} \; \mathrm{false}; \\
& \qquad
  \{ \text{ Tightening: $\bigO(n)$ } \} \\
& \qquad
  \kw{for} \; i := 0 \; \kw{to} \; 2n-1 \; \kw{do} \\
& \qquad\qquad
    w[i,\bari] := \mathrm{floor}\bigl( w[i,\bari] / 2 \bigr); \\
& \qquad
  \{ \text{ Check for $\Zset$-consistency: $\bigO(n)$ } \} \\
& \qquad
    \kw{for} \; i := 0 \; \kw{to} \; 2n-2 \; \kw{step} \; 2 \; \kw{do} \\
& \qquad\qquad
      \kw{if} \; w[i,\bari] + w[\bari,i] < 0
        \; \kw{return} \; \mathrm{false}; \\
& \qquad
  \{ \text{ Strong coherence: $\bigO(n^2)$ } \} \\
& \qquad
  \kw{for} \; i := 0 \; \kw{to} \; 2n-1 \; \kw{do} \\
& \qquad\qquad
    \kw{for} \; j := 0 \; \kw{to} \; 2n-1 \; \kw{do} \\
& \qquad\qquad\qquad
      w[i,j] := \min\bigl(
                      w[i,j],
                      w[i,\bari]/2 + w[\barj,j]/2
                    \bigr); \\
& \qquad
  \kw{return} \; \mathrm{true};
\end{align*}
\caption{A $\bigO(n^3)$ tight closure algorithm for integer coherent graphs}
\label{fig:tight-closure}
\end{figure}
The combination of the results stated in
Theorems~\ref{thm:one-step-tight-closure}
and~\ref{thm:even-closure-consistency-condition}
(together with the well known result for rational consistency)
leads to an $\bigO(n^3)$ tight closure algorithm,
such as that given by the pseudo-code in Figure~\ref{fig:tight-closure},
that
computes the tight closure of any (possibly inconsistent)
coherent integer-weighted graph
returning the Boolean value `$\mathrm{true}$'
if and only if the input graph is $\Zset$-consistent.

\section{Conclusion and Future Work}
\label{sec:conclusion}

We have presented and fully justified an $\bigO(n^3)$ algorithm
that computes the tight closure of a set of integer octagonal
constraints.
The algorithm ---which is based on the extension to integer-weighted
octagonal graphs of the one we proposed for rational-weighted
octagonal graphs~\cite{BagnaraHMZ05,BagnaraHMZ05TR}--- and its
proof of correctness means the issue about
the possibility of computing the tight closure at a computational
cost that is asymptotically not worse than the cost of computing
all-pairs shortest paths is finally closed.

In the field of hardware and software verification, the integrality
constraint that distinguishes integer-weighted from rational-weighted
octagonal graphs can be seen as an abstraction of the more general
imposition of a set of congruence relations.  Such a set can be encoded
by an element of a suitable abstract domain such as the non-relational
congruence domain of \cite{Granger89}
(that is, of the form $x = a \pmod b$),
the  weakly relational \emph{zone-congruence} domain of \cite{Mine02}
(that is, also allowing the form $x - y = a \pmod b$),
the linear congruence domain of \cite{Granger91},
and the more general fully relational \emph{rational grids} domain developed
in~\cite{BagnaraDHMZ07}.
The combination of such domains with the abstract domain
proposed in \cite{BagnaraHMZ05,BagnaraHMZ05TR}
is likely to provide an interesting complexity-precision trade-off.
Future work includes investigating such a combination,
exploiting the ideas presented in this paper.

\newcommand{\noopsort}[1]{}\hyphenation{ Ba-gna-ra Bie-li-ko-va Bruy-noo-ghe
  Common-Loops DeMich-iel Dober-kat Di-par-ti-men-to Er-vier Fa-la-schi
  Fell-eisen Gam-ma Gem-Stone Glan-ville Gold-in Goos-sens Graph-Trace
  Grim-shaw Her-men-e-gil-do Hoeks-ma Hor-o-witz Kam-i-ko Kenn-e-dy Kess-ler
  Lisp-edit Lu-ba-chev-sky Ma-te-ma-ti-ca Nich-o-las Obern-dorf Ohsen-doth
  Par-log Para-sight Pega-Sys Pren-tice Pu-ru-sho-tha-man Ra-guid-eau Rich-ard
  Roe-ver Ros-en-krantz Ru-dolph SIG-OA SIG-PLAN SIG-SOFT SMALL-TALK Schee-vel
  Schlotz-hauer Schwartz-bach Sieg-fried Small-talk Spring-er Stroh-meier
  Thing-Lab Zhong-xiu Zac-ca-gni-ni Zaf-fa-nel-la Zo-lo }

\end{document}

%% file: prooftree.tex
\newdimen\proofrulebreadth \proofrulebreadth=.05em
\newdimen\proofdotseparation \proofdotseparation=1.25ex
\newdimen\proofrulebaseline \proofrulebaseline=2ex
\newcount\proofdotnumber \proofdotnumber=3
\let\then\relax
\def\hfi{\hskip0pt plus.0001fil}
\mathchardef\squigto="3A3B
\newif\ifinsideprooftree\insideprooftreefalse
\newif\ifonleftofproofrule\onleftofproofrulefalse
\newif\ifproofdots\proofdotsfalse
\newif\ifdoubleproof\doubleprooffalse
\let\wereinproofbit\relax
\newdimen\shortenproofleft
\newdimen\shortenproofright
\newdimen\proofbelowshift
\newbox\proofabove
\newbox\proofbelow
\newbox\proofrulename
\def\shiftproofbelow{\let\next\relax\afterassignment\setshiftproofbelow\dimen0 }
\def\shiftproofbelowneg{\def\next{\multiply\dimen0 by-1 }%
\afterassignment\setshiftproofbelow\dimen0 }
\def\setshiftproofbelow{\next\proofbelowshift=\dimen0 }
\def\setproofrulebreadth{\proofrulebreadth}

\def\prooftree{
\ifnum  \lastpenalty=1
\then   \unpenalty
\else   \onleftofproofrulefalse
\fi
\ifonleftofproofrule
\else   \ifinsideprooftree
        \then   \hskip.5em plus1fil
        \fi
\fi
\bgroup
\setbox\proofbelow=\hbox{}\setbox\proofrulename=\hbox{}%
\let\justifies\proofover\let\leadsto\proofoverdots\let\Justifies\proofoverdbl
\let\using\proofusing\let\[\prooftree
\ifinsideprooftree\let\]\endprooftree\fi
\proofdotsfalse\doubleprooffalse
\let\thickness\setproofrulebreadth
\let\shiftright\shiftproofbelow \let\shift\shiftproofbelow
\let\shiftleft\shiftproofbelowneg
\let\ifwasinsideprooftree\ifinsideprooftree
\insideprooftreetrue
\setbox\proofabove=\hbox\bgroup$\displaystyle 
\let\wereinproofbit\prooftree
\shortenproofleft=0pt \shortenproofright=0pt \proofbelowshift=0pt
\onleftofproofruletrue\penalty1
}

\def\eproofbit{
\ifx    \wereinproofbit\prooftree
\then   \ifcase \lastpenalty
        \then   \shortenproofright=0pt  
        \or     \unpenalty\hfil         
        \or     \unpenalty\unskip       
        \else   \shortenproofright=0pt  
        \fi
\fi
\global\dimen0=\shortenproofleft
\global\dimen1=\shortenproofright
\global\dimen2=\proofrulebreadth
\global\dimen3=\proofbelowshift
\global\dimen4=\proofdotseparation
\global\count255=\proofdotnumber
$\egroup  
\shortenproofleft=\dimen0
\shortenproofright=\dimen1
\proofrulebreadth=\dimen2
\proofbelowshift=\dimen3
\proofdotseparation=\dimen4
\proofdotnumber=\count255
}

\def\proofover{
\eproofbit 
\setbox\proofbelow=\hbox\bgroup 
\let\wereinproofbit\proofover
$\displaystyle
}%
\def\proofoverdbl{
\eproofbit 
\doubleprooftrue
\setbox\proofbelow=\hbox\bgroup 
\let\wereinproofbit\proofoverdbl
$\displaystyle
}%
\def\proofoverdots{
\eproofbit 
\proofdotstrue
\setbox\proofbelow=\hbox\bgroup 
\let\wereinproofbit\proofoverdots
$\displaystyle
}%
\def\proofusing{
\eproofbit 
\setbox\proofrulename=\hbox\bgroup 
\let\wereinproofbit\proofusing
\kern0.3em$
}

\def\endprooftree{
\eproofbit 
  \dimen5 =0pt
\dimen0=\wd\proofabove \advance\dimen0-\shortenproofleft
\advance\dimen0-\shortenproofright
\dimen1=.5\dimen0 \advance\dimen1-.5\wd\proofbelow
\dimen4=\dimen1
\advance\dimen1\proofbelowshift \advance\dimen4-\proofbelowshift
\ifdim  \dimen1<0pt
\then   \advance\shortenproofleft\dimen1
        \advance\dimen0-\dimen1
        \dimen1=0pt
        \ifdim  \shortenproofleft<0pt
        \then   \setbox\proofabove=\hbox{%
                        \kern-\shortenproofleft\unhbox\proofabove}%
                \shortenproofleft=0pt
        \fi
\fi
\ifdim  \dimen4<0pt
\then   \advance\shortenproofright\dimen4
        \advance\dimen0-\dimen4
        \dimen4=0pt
\fi
\ifdim  \shortenproofright<\wd\proofrulename
\then   \shortenproofright=\wd\proofrulename
\fi
\dimen2=\shortenproofleft \advance\dimen2 by\dimen1
\dimen3=\shortenproofright\advance\dimen3 by\dimen4
\ifproofdots
\then
        \dimen6=\shortenproofleft \advance\dimen6 .5\dimen0
        \setbox1=\vbox to\proofdotseparation{\vss\hbox{$\cdot$}\vss}%
        \setbox0=\hbox{%
                \advance\dimen6-.5\wd1
                \kern\dimen6
                $\vcenter to\proofdotnumber\proofdotseparation
                        {\leaders\box1\vfill}$%
                \unhbox\proofrulename}%
\else   \dimen6=\fontdimen22\the\textfont2 
        \dimen7=\dimen6
        \advance\dimen6by.5\proofrulebreadth
        \advance\dimen7by-.5\proofrulebreadth
        \setbox0=\hbox{%
                \kern\shortenproofleft
                \ifdoubleproof
                \then   \hbox to\dimen0{%
                        $\mathsurround0pt\mathord=\mkern-6mu%
                        \cleaders\hbox{$\mkern-2mu=\mkern-2mu$}\hfill
                        \mkern-6mu\mathord=$}%
                \else   \vrule height\dimen6 depth-\dimen7 width\dimen0
                \fi
                \unhbox\proofrulename}%
        \ht0=\dimen6 \dp0=-\dimen7
\fi
\let\doll\relax
\ifwasinsideprooftree
\then   \let\VBOX\vbox
\else   \ifmmode\else$\let\doll=$\fi
        \let\VBOX\vcenter
\fi
\VBOX   {\baselineskip\proofrulebaseline \lineskip.2ex
        \expandafter\lineskiplimit\ifproofdots0ex\else-0.6ex\fi
        \hbox   spread\dimen5   {\hfi\unhbox\proofabove\hfi}%
        \hbox{\box0}%
        \hbox   {\kern\dimen2 \box\proofbelow}}\doll%
\global\dimen2=\dimen2
\global\dimen3=\dimen3
\egroup 
\ifonleftofproofrule
\then   \shortenproofleft=\dimen2
\fi
\shortenproofright=\dimen3
\onleftofproofrulefalse
\ifinsideprooftree
\then   \hskip.5em plus 1fil \penalty2
\fi
}


%% file: tc.bbl
\begin{thebibliography}{10}

\bibitem{BagnaraDHMZ07}
R.~Bagnara, K.~Dobson, P.~M. Hill, M.~Mundell, and E.~Zaffanella.
\newblock Grids: A domain for analyzing the distribution of numerical values.
\newblock In G.~Puebla, editor, {\em Logic-based Program Synthesis and
  Transformation, 16th International Symposium}, volume 4407 of {\em Lecture
  Notes in Computer Science}, pages 219--235, Venice, Italy, 2007.
  Springer-Verlag, Berlin.

\bibitem{BagnaraHMZ05}
R.~Bagnara, P.~M. Hill, E.~Mazzi, and E.~Zaffanella.
\newblock Widening operators for weakly-relational numeric abstractions.
\newblock In C.~Hankin and I.~Siveroni, editors, {\em Static Analysis:
  Proceedings of the 12th International Symposium}, volume 3672 of {\em Lecture
  Notes in Computer Science}, pages 3--18, London, UK, 2005. Springer-Verlag,
  Berlin.

\bibitem{BagnaraHMZ05TR}
R.~Bagnara, P.~M. Hill, E.~Mazzi, and E.~Zaffanella.
\newblock Widening operators for weakly-relational numeric abstractions.
\newblock Quaderno 399, Dipartimento di Matematica, Universit\`a di Parma,
  Italy, 2005.
\newblock Available at \url{http://www.cs.unipr.it/Publications/}.

\bibitem{BagnaraHZ06TR}
R.~Bagnara, P.~M. Hill, and E.~Zaffanella.
\newblock The {Parma Polyhedra Library}: Toward a complete set of numerical
  abstractions for the analysis and verification of hardware and software
  systems.
\newblock Quaderno 457, Dipartimento di Matematica, Universit\`a di Parma,
  Italy, 2006.
\newblock Available at \url{http://www.cs.unipr.it/Publications/}. Also
  published as {\tt arXiv:cs.MS/0612085}, available from
  \url{http://arxiv.org/}.

\bibitem{BalasundaramK89}
V.~Balasundaram and K.~Kennedy.
\newblock A technique for summarizing data access and its use in parallelism
  enhancing transformations.
\newblock In B.~Knobe, editor, {\em Proceedings of the ACM SIGPLAN'89
  Conference on Programming Language Design and Implementation (PLDI)}, volume
  24(7) of {\em ACM SIGPLAN Notices}, pages 41--53, Portland, Oregon, USA,
  1989. ACM Press.

\bibitem{BallCLZ04}
T.~Ball, B.~Cook, S.~K. Lahiri, and L.~Zhang.
\newblock {Zapato}: Automatic theorem proving for predicate abstraction
  refinement.
\newblock In R.~Alur and D.~Peled, editors, {\em Computer Aided Verification:
  Proceedings of the 16th International Conference}, volume 3114 of {\em
  Lecture Notes in Computer Science}, pages 457--461, Boston, MA, USA, 2004.
  Springer-Verlag, Berlin.

\bibitem{CormenLR90}
T.~H. Cormen, T.~E. Leiserson, and R.~L. Rivest.
\newblock {\em Introduction to Algorithms}.
\newblock The MIT Press, Cambridge, MA, 1990.

\bibitem{CousotC77}
P.~Cousot and R.~Cousot.
\newblock Abstract interpretation: A unified lattice model for static analysis
  of programs by construction or approximation of fixpoints.
\newblock In {\em Proceedings of the Fourth Annual ACM Symposium on Principles
  of Programming Languages}, pages 238--252, New York, 1977. ACM Press.

\bibitem{CousotCFMMMR05}
P.~Cousot, R.~Cousot, J.~Feret, L.~Mauborgne, A.~Min\'e, D.~Monniaux, and
  X.~Rival.
\newblock The {ASTR\'EE} analyzer.
\newblock In M.~Sagiv, editor, {\em Programming Languages and Systems,
  Proceedings of the 14th European Symposium on Programming}, volume 3444 of
  {\em Lecture Notes in Computer Science}, pages 21--30, Edinburgh, UK, 2005.
  Springer-Verlag, Berlin.

\bibitem{Granger89}
P.~Granger.
\newblock Static analysis of arithmetical congruences.
\newblock {\em International Journal of Computer Mathematics}, 30:165--190,
  1989.

\bibitem{Granger91}
P.~Granger.
\newblock Static analysis of linear congruence equalities among variables of a
  program.
\newblock In S.~Abramsky and T.~S.~E. Maibaum, editors, {\em TAPSOFT'91:
  Proceedings of the International Joint Conference on Theory and Practice of
  Software Development, Volume 1: Colloquium on Trees in Algebra and
  Programming (CAAP'91)}, volume 493 of {\em Lecture Notes in Computer
  Science}, pages 169--192, Brighton, UK, 1991. Springer-Verlag, Berlin.

\bibitem{HarveyS97}
W.~Harvey and P.~J. Stuckey.
\newblock A unit two variable per inequality integer constraint solver for
  constraint logic programming.
\newblock In M.~Patel, editor, {\em ACSC'97: Proceedings of the 20th
  Australasian Computer Science Conference}, volume~19, pages 102--111.
  Australian Computer Science Communications, 1997.

\bibitem{JaffarMSY94}
J.~Jaffar, M.~J. Maher, P.~J. Stuckey, and R.~H.~C. Yap.
\newblock Beyond finite domains.
\newblock In A.~Borning, editor, {\em Principles and Practice of Constraint
  Programming: Proceedings of the Second International Workshop}, volume 874 of
  {\em Lecture Notes in Computer Science}, pages 86--94, Rosario, Orcas Island,
  Washington, USA, 1994. Springer-Verlag, Berlin.

\bibitem{Lagarias85b}
J.~C. Lagarias.
\newblock The computational complexity of simultaneous {Diophantine}
  approximation problems.
\newblock {\em SIAM Journal on Computing}, 14(1):196--209, 1985.

\bibitem{LahiriM05}
S.~K. Lahiri and M.~Musuvathi.
\newblock An efficient decision procedure for {UTVPI} constraints.
\newblock In B.~Gramlich, editor, {\em Frontiers of Combining Systems:
  Proceedings of the 5th International Workshop, FroCoS 2005}, volume 3717 of
  {\em Lecture Notes in Artificial Intelligence}, pages 168--183, Vienna,
  Austria, 2005. Springer-Verlag, Berlin.

\bibitem{Mine01a}
A.~Min\'e.
\newblock A new numerical abstract domain based on difference-bound matrices.
\newblock In O.~Danvy and A.~Filinski, editors, {\em Proceedings of the 2nd
  Symposium on Programs as Data Objects (PADO 2001)}, volume 2053 of {\em
  Lecture Notes in Computer Science}, pages 155--172, Aarhus, Denmark, 2001.
  Springer-Verlag, Berlin.

\bibitem{Mine01b}
A.~Min\'e.
\newblock The octagon abstract domain.
\newblock In {\em Proceedings of the Eighth Working Conference on Reverse
  Engineering (WCRE'01)}, pages 310--319, Stuttgart, Germany, 2001. IEEE
  Computer Society Press.

\bibitem{Mine02}
A.~Min\'e.
\newblock A few graph-based relational numerical abstract domains.
\newblock In M.~V. Hermenegildo and G.~Puebla, editors, {\em Static Analysis:
  Proceedings of the 9th International Symposium}, volume 2477 of {\em Lecture
  Notes in Computer Science}, pages 117--132, Madrid, Spain, 2002.
  Springer-Verlag, Berlin.

\bibitem{Mine05th}
A.~Min\'e.
\newblock {\em Weakly Relational Numerical Abstract Domains}.
\newblock PhD thesis, \'Ecole Polytechnique, Paris, France, March 2005.

\bibitem{Mine06b}
A.~Min{\'e}.
\newblock The octagon abstract domain.
\newblock {\em Higher-Order and Symbolic Computation}, 19(1):31--100, 2006.

\bibitem{NelsonO77}
G.~Nelson and D.~C. Oppen.
\newblock Fast decision algorithms based on {Union} and {Find}.
\newblock In {\em Proceedings of the 18th Annual Symposium on Foundations of
  Computer Science (FOCS'77)}, pages 114--119, Providence, RI, USA, 1977. IEEE
  Computer Society Press.
\newblock The journal version of this paper is \cite{NelsonO80}.

\bibitem{NelsonO80}
G.~Nelson and D.~C. Oppen.
\newblock Fast decision procedures based on congruence closure.
\newblock {\em Journal of the ACM}, 27(2):356--364, 1980.
\newblock An earlier version of this paper is \cite{NelsonO77}.

\bibitem{Pratt77}
V.~R. Pratt.
\newblock Two easy theories whose combination is hard.
\newblock Memo sent to Nelson and Oppen concerning a preprint of their paper
  \cite{NelsonO77}, September 1977.

\bibitem{VenetB04}
A.~Venet and G.~Brat.
\newblock Precise and efficient static array bound checking for large embedded
  {C} programs.
\newblock In {\em Proceedings of the ACM SIGPLAN 2004 Conference on Programming
  Language Design and Implementation (PLDI'04)}, pages 231--242, Washington,
  DC, USA, 2004. ACM Press.

\end{thebibliography}
